\documentclass[letterpaper, 10 pt, conference]{ieeeconf}  
\usepackage{colortbl}
\definecolor{mygray}{gray}{.9}
\usepackage{cite}
\usepackage{color}
\usepackage{amsfonts,amssymb}
\usepackage{bm}
\usepackage{amsmath}
\usepackage{multirow}
\usepackage{algorithmic}
\usepackage{graphicx}
\usepackage{wrapfig}
\usepackage{subfigure}
\usepackage{mathtools}
\usepackage{float}
\usepackage{textcomp}
\usepackage{graphicx} 
\usepackage{epstopdf}
\usepackage{array}
\usepackage[thmmarks,amsmath]{ntheorem}

\newtheorem{proposition}{Proposition}
\newtheorem{assumption}{Assumption}
\newtheorem{lemma}{Lemma}
\newtheorem{remark}{Remark}

\newtheorem{problem}{Problem}

\usepackage[title]{appendix}
\usepackage{cite}
\usepackage{array}
\usepackage{enumerate}
\usepackage{booktabs}
\usepackage{algorithm}
\usepackage{algorithmic}
\usepackage{setspace}
\usepackage{float}
\usepackage{cases}
\usepackage{mathrsfs}
\renewcommand{\QED}{\QEDopen}

\makeatletter\@mparswitchfalse\makeatother\normalmarginpar 
\usepackage[textwidth=1cm, textsize=scriptsize]{todonotes}

\IEEEoverridecommandlockouts                              
\title{\bf Hierarchical Multi-agent Reinforcement Learning for Warehouse Robot Coordination under Communication Loss}

\author{Weihao Sun$^1$, \textit{Student Member, IEEE}, 
Gehui Xu$^2$, \textit{Member, IEEE}, and \\
Andreas A. Malikopoulos$^{1,3}$, \textit{Senior Member, IEEE}
\thanks{This work is supported in part by NSF under Grants CNS-2401007, CMMI-2348381, IIS-2415478, and in part by MathWorks.}
\thanks{$^1$W. Sun and A. A. Malikopoulos are with the Systems Engineering Program, Cornell University, Ithaca, NY 14850 USA. {\tt\small email: \{ws493,amaliko\}@cornell.edu}}
\thanks{$^2$G. Xu is  with the Department of Electrical and Electronic Engineering, Imperial College London, London SW7 2AZ, UK. {\tt\small email: g.xu@imperial.ac.uk}}
\thanks{$^3$A. A. Malikopoulos is with the Applied Mathematics, Systems Engineering, Mechanical Engineering, Electrical \& Computer Engineering, and School of Civil \& Environmental Engineering, Cornell University, Ithaca, NY, USA. (email: \texttt{amaliko@cornell.edu}) }
}

\begin{document}

\maketitle
\thispagestyle{empty}
\pagestyle{empty}

\begin{abstract}
In this paper, we propose a hierarchical multi-agent reinforcement learning framework for coordinating robot teams in warehouse environments under communication loss. We partition the robot team into groups, with centralized coordination within each group and distributed coordination across groups. Each group uses a recurrent predictor to estimate unavailable interaction information due to communication loss. A higher-level policy then generates a compact coordination reference that conditions the local control policy within each group. A predictive safety filter evaluates and modifies the proposed controls when they violate safety constraints. Simulation results show improved task completion under communication loss, reduced communication growth as the team size increases, and safe operation in the tested scenarios.
\end{abstract}

\section{Introduction}
 
Multi-agent robotic systems are increasingly deployed across a wide range of applications, including warehouse and logistics operations\cite{krnjaic2024scalable}, autonomous transportation\cite{dinneweth2022multi,Sun2026CDC,Malikopoulos2020}, and search and rescue\cite{drew2021multi}. For instance, in modern warehouse and logistics systems, mobile robot teams are widely used to transport parcels efficiently across shared workspaces\cite{agaskar2025deepfleet}. Effective collaboration among the robot team can improve task execution quality and reduce reliance on manual handling. In these scenarios, coordination must support efficient and safe task execution despite strong interactions among robots. However, as the team size grows, these interactions require more information exchange and computational effort, making effective coordination more challenging.

In this paper, we consider a large robot team coordinating in a warehouse environment. The robots must jointly transport goods while avoiding collisions and satisfying their motion constraints. 
A centralized approach is a classical solution method that exploits complete system information.
Existing work includes search-based\cite{jia2025conflict} and rolling horizon planning methods\cite{li2021lifelong}, with other works accounting for execution delays\cite{ma2017multi}.
Although such methods have demonstrated strong coordination performance for large warehouse systems, they require fleet-level information aggregation and joint decision computation, which become increasingly demanding as the team size grows \cite{Malikopoulos2021,Malikopoulos2026a}.
Decentralized and learning-based approaches instead allow each robot to make individual decisions. Graph neural network methods, for example, use the communication topology for coordination \cite{li2020graph,li2021message,le2025combining}, while distributed multi-agent reinforcement learning (MARL) methods exploit interaction graphs and local value functions \cite{jing2024distributed}. Related MARL approaches address cooperative decision-making under partial observations \cite{lowe2017multi,zhang2023cooperative,amato2013decentralized}, while methods for predicting and imputing missing observations have been explored in \cite{santos2025centralized}.
However, decentralized coordination becomes more difficult as the team grows and can degrade when shared information is temporarily unavailable. Hierarchical approaches provide an intermediate structure between centralized and decentralized coordination. Liu et al. partition a warehouse into spatial sectors and combine centralized sector-level planning with decentralized local coordination \cite{liu2020prediction}, while hierarchical MARL has been used for warehouse worker agent management \cite{krnjaic2024scalable}. More generally, hierarchical distributed control can separate higher-level multi-agent coordination through reference generation from local control \cite{wang2024hierarchical}. 
However, existing approaches do not directly address how unavailable information can be incorporated into a hierarchical coordination structure under limited communication while maintaining safe robot interactions in warehouse systems.

To address these challenges, we propose a hierarchical MARL framework for warehouse robot coordination under communication loss. The robot team is partitioned into fixed groups, 
with centralized joint coordination within each group, while coordination across groups is distributed through inter-group communication. 
When communication becomes unavailable, a recurrent predictor implemented using a gated recurrent unit (GRU) estimates the missing interaction information from previously received messages. Rather than directly providing individual robot control policies, 
a group-level reference policy maps it to a finite coordination reference, which then conditions the joint control policy within each group.
The hierarchical structure restricts the inter-group information influencing the local control and reduces direct dependence on inter-group shared information.
Before execution, a predictive safety filter uses one-step reachable states to evaluate the proposed joint control and modifies it when necessary. 
A centralized critic uses the common team return during training, while execution requires no global coordinator. We further provide supporting analysis of reference information flow, communication scaling, and the safety filter, and evaluate the proposed framework in warehouse simulation environments under different team sizes, goal assignments, and communication condition.

The remainder of this paper is organized as follows. In Section II, we formulate the warehouse robot coordination problem. In Section III, we present the proposed hierarchical MARL framework, including interaction state prediction, hierarchical control, safety filtering, and centralized training with distributed group execution. Section IV presents the simulation setup and results. In Section V, we highlight the conclusion and future work.

\section{Problem Formulation}

Let
$\mathcal{N}=\{1,\ldots,N\}$ denote a team of $N$ robots operating in the
warehouse.
The robot team is partitioned into $Z$ fixed and non-overlapping control
groups indexed by $\mathcal{Z}=\{1,\ldots,Z\}$.
Let $\mathcal{N}_z\subseteq\mathcal{N}$ denote the set of robots assigned
to group $z\in\mathcal{Z}$, such that
$\mathcal{N}= \bigcup_{z\in\mathcal{Z}}\mathcal{N}_z, 
\mathcal{N}_z\cap\mathcal{N}_{z'} = \emptyset,z\neq z'.$
We consider a fixed group partition over the task horizon, as introduced in the following assumption.
\begin{assumption}
The assignment of robots to control groups is fixed over the task horizon and does not depend on the physical locations of the robots.
\end{assumption}

We consider a grid-based warehouse environment with a feasible position set
$\mathcal{P}\subseteq\mathbb{R}^2$. 
For each robot $i\in\mathcal{N}_z$, we define its state as
\begin{equation}\label{eq:agent_state}
    x_t^{i,z}
    =
    p_t^{i,z},
\end{equation}
where $p_t^{i,z}\in\mathcal{P}$ denotes the robot position and
$t=0,\cdots,T-1$ denotes the time step.
Each robot $i\in\mathcal{N}_z$ is assigned a goal $g^{i,z}\in\mathcal{P}.$

We denote each robot's control input by $u_t^{i,z}\in\mathcal U$. 
For the warehouse setting considered in this paper, we use the discrete control space, \textit{i.e.,} $\mathcal U =
\{\mathrm{Up},\mathrm{Down},\mathrm{Left},\mathrm{Right},\mathrm{Wait}\}.$
Each control command specifies the motion of the robot during one time step, while $\mathrm{Wait}$ keeps the robot at its current position.

For each group $z$, we define the group-level state and control as
$X_t^z=(
x_t^{i,z}:i\in\mathcal{N}_z),
U_t^z
=(u_t^{i,z}:i\in\mathcal{N}_z).$
The corresponding global state and joint control are then denoted by
$X_t=(
X_t^z:z\in\mathcal{Z}),
U_t=(U_t^z:z\in\mathcal{Z}).$
Then, we define the robot team dynamics according to
\begin{equation}\label{eq:system_dynamics}
    X_{t+1}
    =
    F_t
    \left(
    X_t,
    U_t
    \right),
\end{equation}
where $F_t$ denotes the system transition function.

We next specify the state information available to each group coordinator.
\begin{assumption}\label{asp:group_state}
At each time $t$, the coordinator of group $z$ has access to the current state $X_t^z$ of all robots assigned to that group.
\end{assumption}
Under Assumption~\ref{asp:group_state}, coordination is centralized with all information for each robot available within each group. Coordination among different groups relies on communicated information, which may be unavailable due to communication loss. Thus, we introduce the following assumption.

\begin{assumption}\label{asp:communication}
Communication loss affects only information exchanged between different groups.
\end{assumption}
The inter-group communication structure is described by a directed graph $\mathcal{G} = \left( \mathcal{Z}, \mathcal{E}_c \right),$ 
where $(z',z)\in\mathcal{E}_c, z,z'\in \mathcal Z$ indicates that group $z'$ may transmit
coordination information to group $z$. Let
$\mathcal C^z= \{ z'\in\mathcal Z:(z',z)\in\mathcal E_c\}$
denote communication neighbors of group $z$, i.e., the groups that may transmit information to group $z$.

At time $t$, the message transmitted from group $z'$ to group $z$ is denoted by
\begin{equation}\label{eq:message}
    M_t^{z',z}
    =
    \left(
    X_t^{z'},
    U_{t-1}^{z'}
    \right).
\end{equation}
Since the inter-group communication may be unavailable, we then
introduce the availability indicator $A_t^{z',z}\in\{0,1\}$, where $A_t^{z',z}=1$
if the current message is successfully received and $A_t^{z',z}=0$ otherwise.
Then, the information successfully received by group $z$ at time $t$ is collected as
\begin{equation}\label{eq:shared_information}
    S_t^z
    =
    \left\{
    M_t^{z',z}:
    z'\in\mathcal C^z,
    A_t^{z',z}=1
    \right\}.
\end{equation}

The robot team receives a shared reward at each time step that captures progress toward assigned goals, task completion, execution efficiency, and collision avoidance. Specifically, the team reward is defined as
\begin{align}\label{eq:reward}
R_t
=\!
\frac{1}{N}
\!\sum_{z\in\mathcal Z}
\!\sum_{i\in\mathcal N}
\big[
q_p
\big(
\!D(p_t^{i,z},g^{i,z})
\!-
\!D(p_{t+1}^{i,z},g^{i,z})
\big)
\notag\\
\!+\!
q_g \mathbf{1}_t^{i,z,\mathrm{g}}
-\!
q_c \mathbf{1}_t^{i,z,\mathrm{c}}
-
q_{\tau}\mathbf{1}_t^{i,z,\mathrm u}
\big],
\end{align}
where $D(p_t^{i,z},g^{i,z})$ denotes the distance between robot $i$ in group $z$ and its assigned goal, $\mathbf{1}_t^{i,z,\mathrm{g}}=1$ when robot $i$ reaches its goal for the first time at time $t$, $\mathbf{1}_t^{i,z,\mathrm{c}}=1$ if robot $i$ is involved in a collision at time $t$, and $\mathbf{1}_t^{i,z,\mathrm{u}}=1$ if robot $i$ has not yet completed its assigned task at time $t$. $q_p$, $q_g$, $q_c$, and $q_\tau$ denote the positive weight of each term.

For each group $z\in\mathcal Z$, let $\mu_{\psi_z}$ and
$\pi_{\theta_z}$ denote the reference and local control policies,
respectively, where the reference policy generates an abstract
coordination variable that conditions the local control policy, and $\psi_z$ and $\theta_z$ are their corresponding
parameters. Let $\psi=(\psi_z:z\in\mathcal Z)$ and
$\theta=(\theta_z:z\in\mathcal Z)$ be the joint policy parameters
across all groups.
Then, the overall objective is to learn the hierarchical reference and local control policies to maximize the expected cumulative team reward
\begin{equation}\label{eq:objective}
    J(\psi,\theta)
    =
    \mathbb{E}_{\Pi_{\psi,\theta}}
    \left[
    \sum_{t=0}^{T-1}
    R_t
    \right],
\end{equation}
where $\Pi_{\psi,\theta}$ denotes the joint hierarchical policy induced by the reference and local control policies. The coordination problem considered in this paper is stated as follows.

\begin{problem}
Determine the reference and local control policies for each group that maximize the expected finite-horizon team reward in \eqref{eq:objective} under the system dynamics \eqref{eq:system_dynamics} and inter-group communication loss.
\end{problem}

\section{Hierarchical Coordination Algorithm}

In this section, we introduce the proposed hierarchical MARL warehouse coordination framework. Fig.~\ref{fig:algorithm_flow} illustrates the algorithm flow, combining recurrent interaction state prediction, hierarchical control for both group-level and local robots, a safety filter, and distributed group execution with centralized training. Each component will be introduced in detail in the following subsections. 

\begin{figure*}[]
    \centering
    \includegraphics[width=0.9\textwidth,height=0.3\textheight]{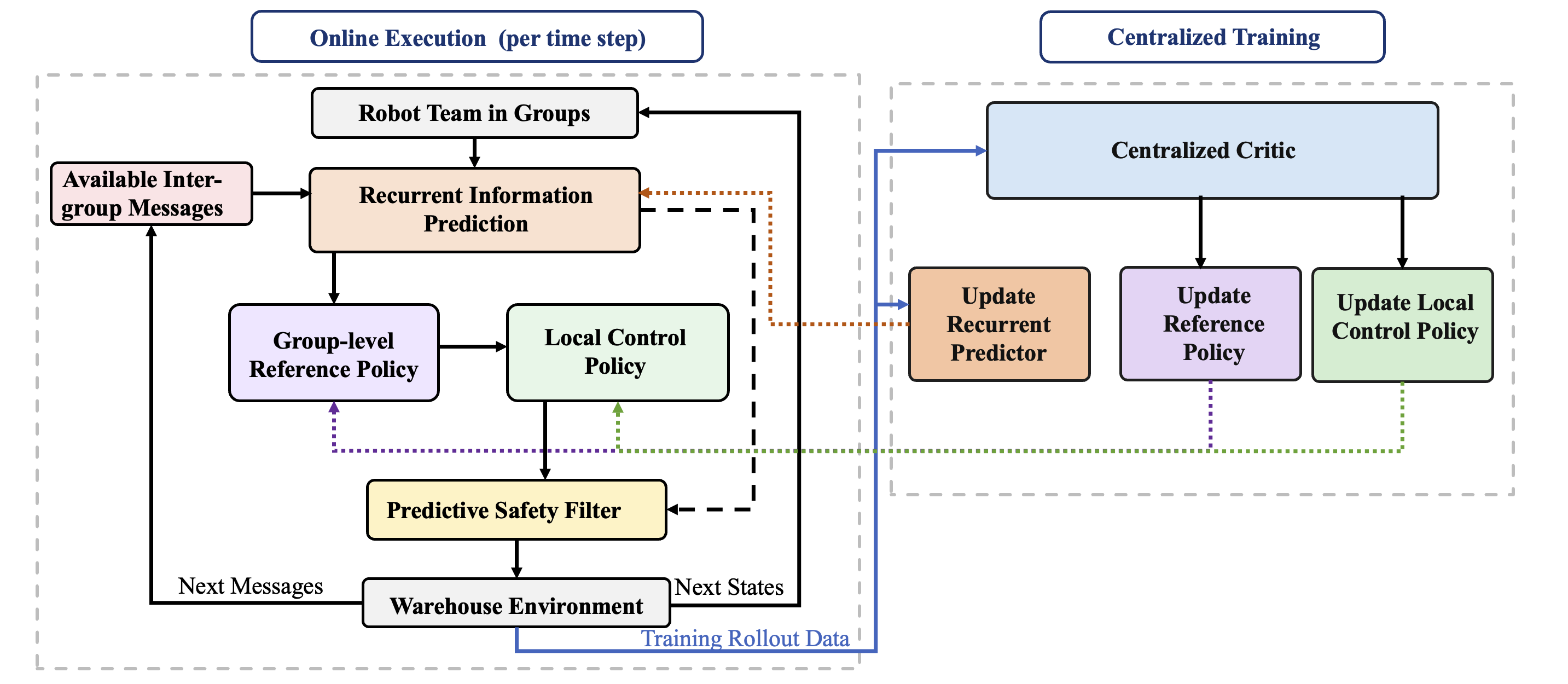}
    \vspace{-0.4cm}
    \caption{Overview of the proposed hierarchical coordination framework.}
    \label{fig:algorithm_flow}
\end{figure*}

\subsection{Inter-group Information Prediction}

The recurrent interaction state predictor estimates interaction information unavailable due to inter-group communication loss. We implement the predictor using a GRU so that previously received communication information can be retained and used to estimate the missing current information.

We define the predictor
input as
\begin{equation}\label{eq:predictor_input}
    \overline M_t^{z',z}
    =
    \begin{cases}
        M_t^{z',z},
        & A_t^{z',z}=1,\\
        \varnothing,
        & A_t^{z',z}=0.
    \end{cases}
\end{equation}
where $\varnothing$ denotes a null input with the same fixed dimension as the encoded message.

The recurrent predictor updates its hidden state and generates predicted information according to
\begin{equation}\label{eq:recurrent_predictor}
    \left(
    h_t^{z',z},
    \widetilde M_t^{z',z}
    \right)
    =
    f_{\eta_z}
    \left(
    h_{t-1}^{z',z},
    \overline M_t^{z',z},
    X_t^z,
    A_t^{z',z}
    \right),
\end{equation}
where  $\widetilde M_t^{z',z}$ denotes the predicted information, $f_{\eta_z}$ is a recurrent predictor parameterized by
$\eta_z$, and $h_t^{z',z}$ denotes the GRU hidden state that encodes previously available inter-group information associated with group $z'$.
If communication from group $z'$ is available, the received message is used directly. Otherwise, the prediction is used. Specifically,
\begin{equation}\label{eq:interaction_state_estimate}
    \widehat{M}_t^{z',z}
    =
    \begin{cases}
        M_t^{z',z},
        & A_t^{z',z}=1,\\[1mm]
        \widetilde M_t^{z',z},
        & A_t^{z',z}=0,
    \end{cases}
\end{equation}
Collecting the completed messages from all communication neighbors,
the inter-group information used by group $z$ is
\begin{equation}\label{eq:interaction_state_collection}
    \widehat S_t^z
    =
    \left(
    \widehat M_t^{z',z}:
    z'\in\mathcal C^z
    \right).
\end{equation}

During centralized training, the training process retains the ground-truth message $M_t^{z',z}$ as a supervised target, regardless of the value of $A_t^{z',z}$. We then define the loss function of the recurrent predictor for each group $z\in\mathcal Z$ as:
\begin{equation}\label{eq:gru_pred_loss}
    \mathcal L_{\mathrm{c}}(\eta_z) = 
    \mathbb E \left[
    \sum_{z'\in\mathcal C^z} 
    \left( 1-A_t^{z',z} \right) 
    \left\| \widetilde M_t^{z',z} - M_t^{z',z} \right\|_2^2
    \right],
\end{equation}
where $1-A_t^{z',z}$ ensures prediction error is included only when the current message from group $z'$ is unavailable to group $z$. The predictor parameters $\eta_z$ are updated during centralized training using this loss.

Beyond mitigating missing information, the group-based communication structure also reduces the repeated information transmission across the network. To illustrate, consider a fully decentralized reference case where all robots share their information with every other robot. Let $d_m$ denote the dimension of the complete communicated information associated with one robot. Such a structure requires $N(N-1)d_m$ communicated scalar elements at each time step. However, in the proposed architecture, each robot's information is incorporated once within its group and is subsequently shared at most once with each of the other $Z-1$ groups. Therefore, even when the group-level messages collectively contain information from all robots, the communication load satisfies
\begin{equation}
    C_{\mathrm{hier}} \leq N(Z-1)d_m,
    \qquad
    \frac{C_{\mathrm{hier}}}{C_{\mathrm{dec}}}
    \leq \frac{Z-1}{N-1},
\end{equation}
where $C_{\mathrm{dec}}=N(N-1)d_m$ denotes the fully decentralized all to all case. Hence, when $Z\ll N$, the hierarchical structure reduces repeated robot-level information sharing from $\mathcal O(N^2)$ to $\mathcal O(NZ)$. In particular, for a fixed number of groups, the repeated communication load grows linearly with $N$, compared with quadratic growth for robot-level all-to-all communication.

\subsection{Inter-group Reference Generator}

Using the completed inter-group information $\widehat S_t^z$ constructed in the previous subsection, we use the reference layer to generate a reduced coordination reference for each group. For group $z\in\mathcal Z$, the reference policy is conditioned on the current group state $X_t^z$, the assigned goal information, and $\widehat S_t^z$. For robots in each group $z$, we collect the assigned goals in
$G^z =\left(g^{i,z}: i\in\mathcal N_z\right),G^z\in\mathcal P^{|\mathcal N_z|}$. We include this information in the policy observations to ensure the reference and control policies account for the assigned goals.

For each group $z\in\mathcal Z$, we define the reference observation
\begin{equation}
    o_{t}^z = \operatorname{Enc}_r
    \left( X_t^z, G^z, \widehat S_t^z
    \right),
\end{equation}
where $\operatorname{Enc}_r$ maps the available information to a fixed-dimensional representation used by the reference policy.
The reference policy then generates
\begin{equation}\label{eq:reference_policy}
    r_t^z
    \sim
    \mu_{\psi_z}
    \left(
    \cdot
    \mid o_{t}^z
    \right),\qquad
    r_t^z\in\mathcal R_z,
\end{equation}
where $\psi_z$ are the policy parameters, and $\mathcal R_z$ is a finite discrete reference space whose elements are abstract coordination variables without predefined physical meaning.

During execution, a deterministic reference can be selected according to
\begin{equation}
    r_t^z \in \arg\max_{r\in\mathcal R_z} \mu_{\psi_z}
    \left( r \mid o_{t}^z \right).
\end{equation}
Notice that the reference policy does not directly specify the control of individual robots. Instead, it provides a reduced coordination variable that conditions the subsequent local control policy, thereby separating group coordination from detailed joint robot control.

\subsection{Local Control}

The local control layer then converts the group reference $r_t^z$ into a proposed joint control for all robots assigned to group $z\in\mathcal Z$. Unlike the reference policy, which incorporates intergroup information to determine a coordination command, the local control policy operates only on the complete group state $X^z_t$, goal information $G^z$, and the generated reference $r_t^z$. The local policy does not directly receive $\widehat S_t^z$. Thus, inter-group communication only influences the proposed local control through the group-level reference.

For each group $z$, we define the local control observation: 
\begin{equation}
    w_{t}^z = 
    \operatorname{Enc}_u 
    \left( X_t^z, G^z, r_t^z \right),
\end{equation}
where $\operatorname{Enc}_u$ denotes another encoder that transforms the information into the input representation used by the local actor. 
Then, the local control policy generates a proposed
joint control action for each group:
\begin{equation}\label{eq:local_policy}
    \widetilde U_t^z
    \sim
    \pi_{\theta_z}
    \left(
    \cdot
    \mid
    w_t^z
    \right),
    \qquad
    \widetilde U_t^z\in \mathcal U^{|\mathcal N_z|},
\end{equation}
where $\theta_z$ denotes the parameters of the local control policy.

During training, $\widetilde U_t^z$ is sampled from $\pi_{\theta_z}$ according to \eqref{eq:local_policy}. During execution, the corresponding deterministic control can be selected according to
\begin{equation}\label{eq:local_deterministic_control}
    \widetilde U_t^z \in 
    \arg\max_{U\in\mathcal U^{|\mathcal N_z|}}
    \pi_{\theta_z} \left( U \mid w_{t}^z \right).
\end{equation}

Under this hierarchical structure, the completed inter-group information $\widehat S^z_t$ does not directly enter the local control policy, but only influences it through reference $r_t^z$. We next characterize this information flow using conditional mutual information. We define the conditional Shannon entropy and conditional mutual information \cite{cover1991elements} as
\begin{align*}
&H\!(r_t^z\!\mid \!X_t^z,\!G^z)
=
-\mathbb E\!\big[
\sum_{r\in\mathcal R_z}
p(r\mid\! X_t^z,G^z)\log p(r \! \mid \!X_t^z,G^z)
\big], \notag\\
&I\!(\widehat S_t^z;\!r_t^z\!\mid \!X_t^z,\!G^z)
=
H\!(r_t^z\mid \!X_t^z,G^z)
-
H\!(r_t^z\mid \!\widehat S_t^z,X_t^z,G^z),
\end{align*}
where $H\!(\cdot),I(\cdot) \in\mathbb R_{\geq 0}$, and $p(\cdot|\cdot)\in[0,1]$ denotes the corresponding conditional
probability mass function. We then introduce the following proposition to bound the information amount.

\begin{proposition}
\label{prop:reference}
Under the reference policy in \eqref{eq:reference_policy} and the
local control policy in \eqref{eq:local_policy}, conditioned on
$X_t^z$ and $G^z$, the completed inter-group information
$\widehat S_t^z$ and the proposed local control $\widetilde U_t^z$
satisfy
\begin{equation*}
I\!\big(
\widehat S_t^z;
\widetilde U_t^z
\mid X_t^z,G^z
\big)
\leq
H\!\big(
r_t^z
\mid X_t^z,G^z
\big)
\leq
\log |\mathcal R_z|.
\end{equation*}

\end{proposition}

\begin{proof}
From the hierarchical structure, conditioned on $X_t^z$ and $G^z$,
the information $\widehat S_t^z$ influences $\widetilde U_t^z$ only
through $r_t^z$. Hence, by the conditional data processing inequality,
$I(\widehat S_t^z;\widetilde U_t^z\mid X_t^z,G^z)
\leq I(\widehat S_t^z;r_t^z\mid X_t^z,G^z)$. By the definition above,
$I(\widehat S_t^z;r_t^z\mid X_t^z,G^z)
=H(r_t^z\mid X_t^z,G^z)
-H(r_t^z\mid \widehat S_t^z,X_t^z,G^z)
\leq H(r_t^z\mid X_t^z,G^z)$, since conditional entropy is nonnegative. Finally, since $r_t^z\in\mathcal R_z$ and $\mathcal R_z$ is finite,
its conditional entropy is maximized by the uniform distribution over
$\mathcal R_z$, yielding $H(r_t^z\mid X_t^z,G^z)\leq\log|\mathcal R_z|$.
\end{proof}

Proposition~\ref{prop:reference} shows that the finite reference restricts the inter-group information influencing the proposed local control to at most $\log|\mathcal R_z|$. Rather than conditioning the local policy directly on the complete inter-group information $\widehat S_t^z$, the hierarchical structure channels this information through the finite reference, thereby bounding it in terms of the reference space cardinality. This mechanism introduces a trade-off: a smaller reference space provides stronger compression but may discard useful information, whereas a larger reference space may preserve more information at the cost of weaker compression. Thus, the choice and cardinality of the reference space constitute important design considerations and remain topics for future investigation.
\begin{remark}
Proposition~\ref{prop:reference} applies to the learned hierarchical policy before safety filtering. The executed control $U_t^z$ may additionally depend on $\widehat S_t^z$ through the safety filter.
\end{remark}

Since the proposed control $\widetilde U_t^z$ is generated by the
learned policy without explicit enforcement of the physical and
operational constraints of the warehouse, it is passed through a
safety filter before execution, which is introduced in the following subsection.

\subsection{Safety Filter}
The local policy does not execute the proposed joint control directly. Instead, each group applies a predictive safety filter that evaluates the proposed control using a one-step prediction of the robot motion.
At each time step, the filter uses the current group state
$X_t^z$, the proposed joint control $\widetilde U_t^z$, and the
completed inter-group information $\widehat S_t^z$ to evaluate potential
conflicts at the next time step. The filter first constructs the possible
one-step states of robots in neighboring groups. It then examines candidate joint controls for group $z$ in non-decreasing order of their deviation from $\widetilde U_t^z$, predicts the resulting next state of group $z$, and selects the first candidate satisfying the safety conditions. If no
examined candidate is admissible, a predefined fallback control is applied.

The filter first uses $\widehat S_t^z$ to construct one-step predictions of robots in neighboring groups. For each neighboring group $z'\in\mathcal C^z$ and robot $j\in\mathcal N_{z'}$, let $\widehat p_t^{j,z'}$ denote its received or predicted current position contained in $\widehat X_t^{z'}$. Since the current control
$U_t^{z'}$ is not assumed to be available to other groups, the corresponding one-step reachable state set is defined as
\begin{equation}\label{eq:interaction_prediction}
\Omega_{t+1}^{j,z'}
=
\left\{
F_t^{j,z'}\!
\left(
\widehat p_t^{j,z'},u'\!
\right)
:
u'\!\in\mathcal \!U
\right\},
j\in\mathcal N_{z'},
z'\in\mathcal C^z,
\end{equation}
where $\Omega_{t+1}^{j,z'}\subseteq\mathcal P$ contains the possible
next positions of robot $j$ in group $z'$, and $F_t^{j,z'}$ denotes
the corresponding robot component of the system transition in
\eqref{eq:system_dynamics}.

After constructing the reachable sets of the neighboring robots,
the filter generates candidate joint controls for any group $z\in Z$. For a candidate joint control $U=(u^{i,z}:i\in\mathcal N_z)$ and the proposed joint control $\widetilde U_t^z=(\widetilde u_t^{i,z}:i\in\mathcal N_z)$, we define the control deviation
\begin{equation}\label{eq:safety_deviation}
D^z(U, \widetilde U_t^z ) 
= \sum\nolimits_{i\in\mathcal N_z} 
\mathbf 1 
\left\{ u^{i,z} \neq \widetilde u_t^{i,z} \right\},
\end{equation}
where $\mathbf 1\{\cdot\}$ denotes the indicator function, $D^z(U,\widetilde U_t^z)\in\{0,\cdots,|\mathcal N_z|\}$. This deviation therefore measures the number of individual robot controls modified relative to the policy output. The candidate search is controlled by two parameters.
The search radius $L\in\{0,\ldots,|\mathcal N_z|\}$ specifies the
maximum allowable deviation, i.e., $D^z(U,\widetilde U_t^z)\leq L$. The candidate budget $C_{\max}\in\mathbb N$ specifies the maximum number of candidate controls examined at each time step.  
Candidate controls are evaluated in non-decreasing order of $D^z(U,\widetilde U_t^z)$ until an admissible candidate is found, or the budget is exhausted.

For each candidate joint control $U$, the filter predicts the next state
of group $z\in Z$ according to
\begin{equation}\label{eq:safety_prediction}
\widehat X_{t+1}^z(U)
=
F_t^z
\left(
X_t^z,
U
\right),
\end{equation}
where $F_t^z$ denotes the system transition in \eqref{eq:system_dynamics} restricted to the robots in group $z$. Thus, $\widehat X_{t+1}^z(U)$ is the predicted state resulting from applying candidate $U$ to the current group state $X_t^z$. 

Each candidate control $U$ is then evaluated by comparing the predicted next state of group $z$ against the one-step safety conditions defined below. Let $\widehat p_{t+1}^{i,z}(U)$ denote the position component of
$\widehat X_{t+1}^z(U)$ corresponding to robot $i\in\mathcal N_z$. A candidate control $U$ is admissible if and only if
\begin{align}\label{eq:safety_conditions}
\widehat p_{t+1}^{i,z}(U)
&\in\mathcal P,
&& \forall i\in\mathcal N_z,
\\
\widehat p_{t+1}^{i,z}(U)
&\neq
\widehat p_{t+1}^{j,z}(U),
&& \forall i,j\in\mathcal N_z,\ i\neq j,
\notag\\
\big(
\widehat p_{t+1}^{i,z}(U),
\widehat p_{t+1}^{j,z}(U)
\big)
&\neq
\big(
p_t^{j,z},
p_t^{i,z}
\big),
&& \forall i,j\in\mathcal N_z,\ i\neq j,
\notag\\
\left\|
\widehat p_{t+1}^{i,z}(U)-p'
\right\|_2
&>0,
&&
\substack{
\forall z'\in\mathcal C^z,\;
\forall i\in\mathcal N_z,\\
\forall j\in\mathcal N_{z'},\;
\forall p'\in\Omega_{t+1}^{j,z'}.
}\notag
\end{align}
The first condition requires the predicted position of each robot in
group $z$ to remain in the feasible warehouse region. The second and
third conditions prevent collisions within one group $z$ by excluding
simultaneous occupancy of the same position and pairwise position
exchanges. The final condition prevents a robot in group $z$ from occupying the same position as any possible next position of a robot in a neighboring group. The first candidate satisfying all conditions is selected as the executed control. The search terminates immediately after such a candidate is found. If no examined candidate is admissible within the search radius and candidate budget, the filter applies the joint \textsc{Wait} control as the fallback action.

The following result characterizes a sufficient candidate budget for examining all controls within radius $L$. Let $n_z=|\mathcal N_z|$ denote the number of robots in group $z$, and $m=|\mathcal U|$ denote the cardinality of the individual robot control set.

\begin{lemma}\label{lemma:safety_bound}
    Suppose an admissible candidate $U$ exists such that $D^z(U,\widetilde U_t^z)\leq L$. If
    \begin{equation}
        C_{\max} \geq \sum_{d=0}^L \binom{n_z}{d}(m-1)^d,
    \end{equation}
    then the filter examines sufficiently many candidates to find an
    admissible control within the search radius.
\end{lemma}
\begin{proof}
    For a candidate differing from $\widetilde U_t^z$ in exactly $d$ robot controls, there are $\binom{n_z}{d}$ ways to select the modified robots and $(m-1)^d$ possible alternative control assignments. Therefore, the number of candidates with deviation at most $L$ is $\sum_{d=0}^{L}\binom{n_z}{d}(m-1)^d$. Hence, if $C_{\max}$ is no smaller than this quantity, every candidate within the prescribed search radius can be examined, and any admissible candidate in this region is guaranteed to be found.
\end{proof}

Lemma~\ref{lemma:safety_bound} shows that the search radius $L$ controls
the tradeoff between correction capability and computational effort.
Increasing $L$ enlarges the candidate set available to the filter, while
decreasing $L$ restricts the search to controls closer to the learned
proposal and reduces the search effort. We next characterize the one-step safety property provided by the reachable state sets of the interacting groups.

\begin{proposition}
\label{prop:one_step_safety}
Suppose for each $z'\in\mathcal C^z$ and $j\in\mathcal N_{z'}$, $F_t^{j,z'}$ matches the corresponding system transition, and $\widehat p_t^{j,z'}=p_t^{j,z'}$. If a candidate control $U$ satisfies
\eqref{eq:safety_conditions}, then no inter-group collision occurs at
time $t+1$ for any
$U_t^{z'}\in\mathcal U^{|\mathcal N_{z'}|}$.
\end{proposition}

\begin{proof}
For any $U_t^{z'}\in\mathcal U^{|\mathcal N_{z'}|}$ and each
$j\in\mathcal N_{z'}$, let $u_t^{j,z'}$ denote the corresponding
individual robot control. From \eqref{eq:interaction_prediction},
$p_{t+1}^{j,z'}=F_t^{j,z'}\big(p_t^{j,z'},u_t^{j,z'}\big)\in\Omega_{t+1}^{j,z'},$ since $\widehat p_t^{j,z'}=p_t^{j,z'}$.
Because \eqref{eq:safety_conditions} is required to hold for every
$p'\in\Omega_{t+1}^{j,z'}$, it holds in particular for the actual
position $p_{t+1}^{j,z'}$. Hence, the candidate $U$ avoids an inter-group collision at time $t+1$.
\end{proof}

\begin{remark}
Under communication loss, $\Omega_{t+1}^{j,z'}$ is constructed from
the predicted position $\widehat p_t^{j,z'}$. Therefore, the resulting safety condition depends on the accuracy of this prediction.
\end{remark}

The complete safety filtering procedure is summarized in
Algorithm~\ref{alg:safety_filter}.

\begin{algorithm}[t]
\caption{Predictive Safety Filter}
\label{alg:safety_filter}
\begin{algorithmic}[1]
\STATE Receive $\widetilde U_t^z$, $X_t^z$, and $\widehat S_t^z$.
\STATE Construct $\{\Omega_{t+1}^{j,z'}\}_{j\in\mathcal N_{z'},\,z'\in\mathcal C^z}$ using \eqref{eq:interaction_prediction}.
\STATE Generate candidate controls in nondecreasing order of $D^z(U,\widetilde U_t^z)$, subject to $L$ and $C_{\max}$.
\FOR{each generated candidate $U$} 
    \STATE Predict $\widehat X_{t+1}^z(U)$ using \eqref{eq:safety_prediction}.
    \STATE Check the admissibility conditions in
    \eqref{eq:safety_conditions}.
    \IF{the candidate is admissible}
        \STATE Set $U_t^z=U$ and terminate the search.
    \ENDIF
\ENDFOR
\IF{no examined candidate is admissible}
\STATE Set $U_t^z=(\textsc{Wait}:i\in\mathcal N_z)$.
\ENDIF
\RETURN $U_t^z$.
\end{algorithmic}
\end{algorithm}

\subsection{Centralized Training and Distributed Group Execution}
The proposed framework follows the centralized training with distributed execution paradigm \cite{lowe2017multi,foerster2018counterfactual}. During training, information from all groups is available to a centralized critic, while trajectories are generated by the hierarchical policies together with the safety filter. During execution, each group coordinator independently computes a joint control for its assigned robots using its group state and available inter-group information, without a global coordinator.

For the resulting trajectory, the return from time $t$ is defined as
\begin{equation}\label{eq:training_return}
G_t = \sum\nolimits_{k=t}^{T-1}R_k. 
\end{equation} 
Using the centralized value critic $V_\omega$, the team advantage is
estimated as
\begin{equation}\label{eq:advantage} 
\widehat A_t = G_t - V_\omega(X_t). 
\end{equation}
The same team advantage is then used to update both levels of the hierarchical policy, while the centralized critic is trained using the realized return:
\begin{align}\label{eq:losses}
\mathcal L_{\mathrm{r}}(\psi_z)
&=
-\mathbb E\left[
\widehat A_t
\log\mu_{\psi_z}
\big(r_t^z\mid o_{t}^z\big)
\right],
\notag \\
\mathcal L_{\mathrm{l}}(\theta_z)
&=
-\mathbb E\left[
\widehat A_t
\log\pi_{\theta_z}
\big(\widetilde U_t^z\mid w_t^z\big)
\right],
\notag \\
\mathcal L_V(\omega)
&=
\mathbb E\left[
\left(V_\omega(X_t)-G_t\right)^2
\right].
\end{align}
where $\mathcal L_{\mathrm{r}}(\psi_z)$ and
$\mathcal L_{\mathrm{l}}(\theta_z)$ denote the reference and local
actor losses for group $z$, respectively, and
$\mathcal L_V(\omega)$ denotes the centralized critic loss. The parameters $\{\psi_z,\theta_z,\eta_z\}_{z=1}^{Z}$ and $\omega$
are optimized using their corresponding losses in
\eqref{eq:losses} and
$\mathcal L_{\mathrm{c}}$, respectively.
Notice that the GRU predictors are trained separately from the policy gradient updates using the supervised loss $\mathcal L_c$. During centralized training, the true messages $M_t^{z',z}$ remain available as supervision regardless of communication loss.


Thus, the safety filter is treated as part of the closed-loop transition. Its corrections affect the resulting trajectory and team return, while the policy gradient is evaluated on the sampled policy outputs.

After training, the centralized critic is discarded, and the learned
policy and predictor parameters are fixed. Each group then executes
independently using its recurrent predictor, reference policy, local
control policy, and safety filter, without access to the global state.

\subsection{Algorithm Summary}
The hierarchical MARL coordination algorithm is summarized in Algorithm~\ref{alg:hierarchical_marl}.

\begin{algorithm}[h]
\caption{Hierarchical Warehouse Robot Coordination}
\label{alg:hierarchical_marl}
\begin{algorithmic}[1]
\STATE Initialize fixed robot groups $\{\mathcal{N}_z\}_{z=1}^{Z}$, reference policies $\{\mu_{\psi_z}\}_{z=1}^{Z}$, and local control policies $\{\pi_{\theta_z}\}_{z=1}^{Z}$.
\STATE Initialize recurrent interaction state predictors $\{f_{\eta_z}\}_{z=1}^{Z}$. \STATE Initialize centralized critic $V_{\omega}$.

\FOR{episode $=1,\ldots,M$}
\STATE Reset the environment and recurrent predictor states.
\FOR{$t=0,\ldots,T-1$}
\FOR{each group $z\in\mathcal Z$}
\STATE Observe $X_t^z$, receive $S_t^z$ and
$\{A_t^{z',z}\}_{z'\in\mathcal C^z}$.
\STATE Update GRU and construct $\widehat S_t^z$ using  \eqref{eq:recurrent_predictor} -- \eqref{eq:interaction_state_collection}.
\STATE Sample
$r_t^z\sim\mu_{\psi_z}(\cdot\mid o_{t}^z)$.
\STATE Sample
$\widetilde U_t^z\sim\pi_{\theta_z}(\cdot\mid w_t^z)$.
\STATE Apply Algorithm~\ref{alg:safety_filter} to obtain $U_t^z$.
\ENDFOR
\STATE Execute $U_t=(U_t^z:z\in\mathcal Z)$, observe
$X_{t+1}$ and $R_t$, and store the rollout data.
\ENDFOR
\STATE Update centralized critic $V_{\omega}$.
\STATE Update $\{\psi_z,\theta_z\}_{z=1}^{Z}$ using the common advantage.
\STATE Update $\{\eta_z\}_{z=1}^{Z}$ using $\mathcal L_{\mathrm c}$.
\ENDFOR
\end{algorithmic}
\end{algorithm}

\section{Experiments}
We evaluate the proposed framework in a warehouse navigation environment with fixed shelves, walls, open aisles, and intersecting robot flows. The simulation experiments are designed to examine the performance and robustness of the proposed solution.

\subsection{Simulation Setup}
The simulations are conducted in warehouse layouts consisting of fixed shelf obstacles, traversable aisles, and shared crossing regions. For each configuration, multiple scenarios are generated by randomly sampling feasible initial and goal positions while keeping the warehouse layout fixed. These assignments create diverse traffic flows involving interactions among robots within the same control group and across different groups. Group assignments remain fixed throughout each episode and are independent of the robots' physical locations. 
For each configuration, a separate set of hierarchical policy parameters is trained. Communication loss is introduced during training and evaluation according to the specified communication loss probability $p$. The learned policies are then evaluated with fixed parameters on independently generated scenarios under the corresponding configuration. Each episode terminates when all robots reach their assigned goals or when the prescribed horizon is reached.

We consider robot teams of $20$, $40$, $60$, $80$, $90$, and $100$ robots. The robots are partitioned into $Z=4$ fixed groups. Communication loss probabilities of $p\in\{0.2,0.4,0.6\}$ are considered. Each training run consists of $2000$ episodes, with a maximum horizon of $500$ time steps per episode. For the safety filter, we set $L=2$ and $C_{\max}=500$.
Performance is evaluated using the fraction of robots reaching their assigned goals, accumulated team return, collision rate, communication payload, and safety correction rate. The success rate is defined as the fraction of robots that reach their assigned goals within the episode horizon. The safety correction rate is defined as the fraction of individual robot controls for which the executed control differs from the control proposed by the local policy. Each trained policy is evaluated over $100$ randomly generated episodes. 

For comparison, we consider two baseline methods. The first is a fully decentralized MARL baseline, in which each robot independently selects its control using robot-level observations and communicates with all others without the proposed group-level hierarchy. The second is IC3Net~\cite{singh2018learning}, implemented in the same warehouse environment. Both baselines are trained separately under matched reward, task assignment, robot team, and communication loss settings.

\subsection{Simulation Results}
We assess the proposed framework from multiple perspectives, including training convergence, task completion, communication load, and safety performance.
We first examine the training convergence of the proposed method. Fig.~\ref{fig:fig1}.(a) compares the proposed method with the fully decentralized baseline for teams of $20,40$, and $60$ robots. For $20$ and $40$ robots, both methods can improve during training, while the proposed method reaches a higher final episode reward. However, when the robot team is larger, i.e., $60$ in this comparison case, the fully decentralized baseline remains at a considerably lower reward, while the proposed method continues to improve and converges to a higher reward. Fig.~\ref{fig:fig1}.(b) further reports the training results of the proposed method for larger teams of  $80$, $90$, and $100$ robots. The episode reward increases consistently and stabilizes for all configurations, and converges at around $1600$ episodes.
\begin{figure}
    \centering
    \includegraphics[width=\linewidth]{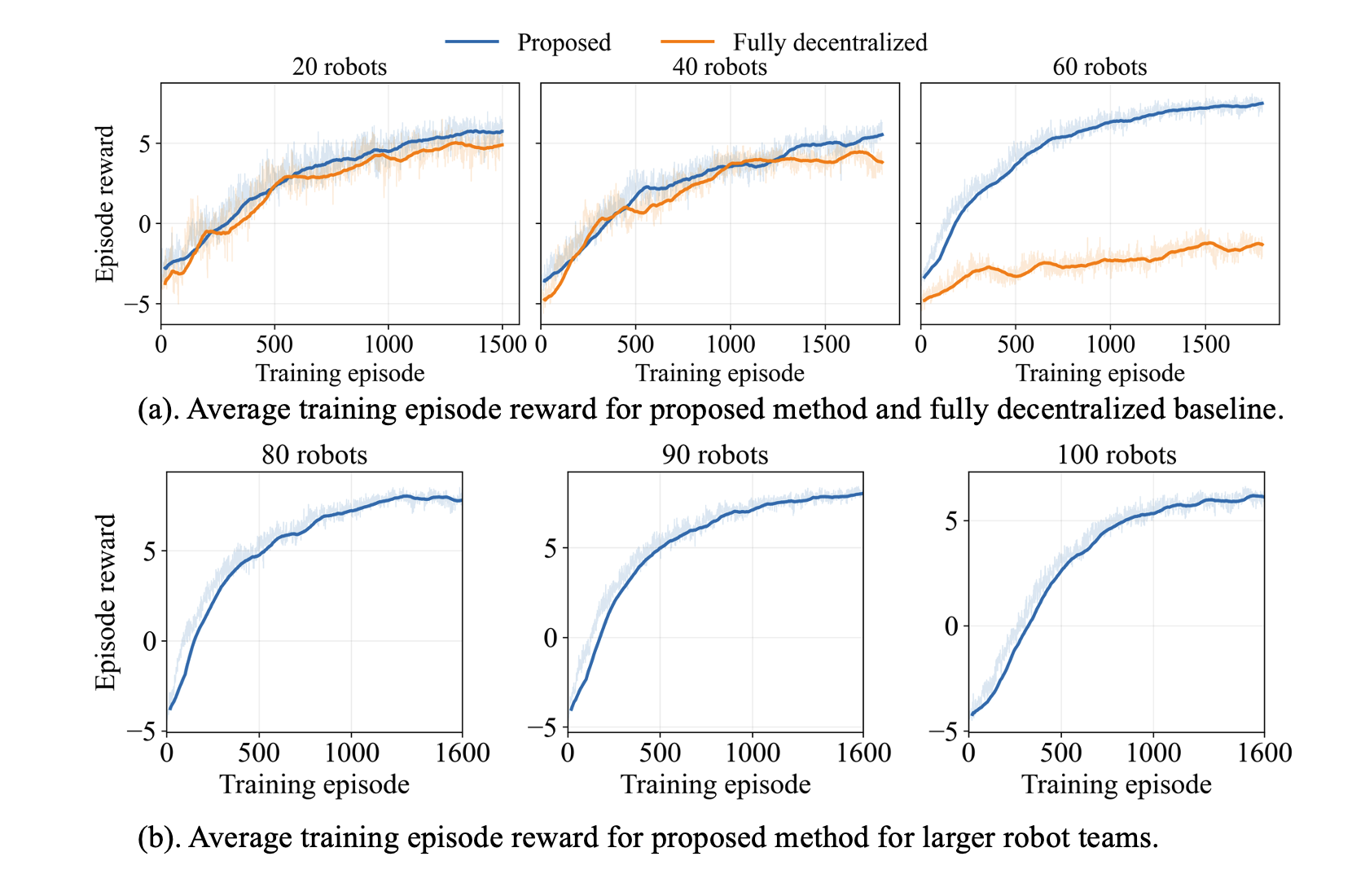}
    \vspace{-1cm}
    \caption{Training convergence of the proposed method and the comparison baseline under different robot team sizes, with a 0.2 communication loss rate.}
    \label{fig:fig1}
\end{figure}

We then evaluate the task completion rate under communication loss. We compare the proposed method with IC3Net for team sizes of $40$ and $60$ robots and communication loss probabilities $p\in\{0.2,0.4,0.6\}$. Fig.~\ref{fig:fig2} shows the comparison results. To visualize this, for each configuration, we sort the evaluation episodes by achieved success rate, and the horizontal axis represents the corresponding episode percentile. The proposed method achieves higher success rates across the entire range of evaluation episodes for all tested cases. This difference may be attributed to the dense robot interactions in shared aisles and intersections, where the hierarchical method with a recurrent predictor provides additional structure for handling communication loss.

\begin{figure}
    \vspace{-0.3cm}
    \centering
    \includegraphics[width=\linewidth]{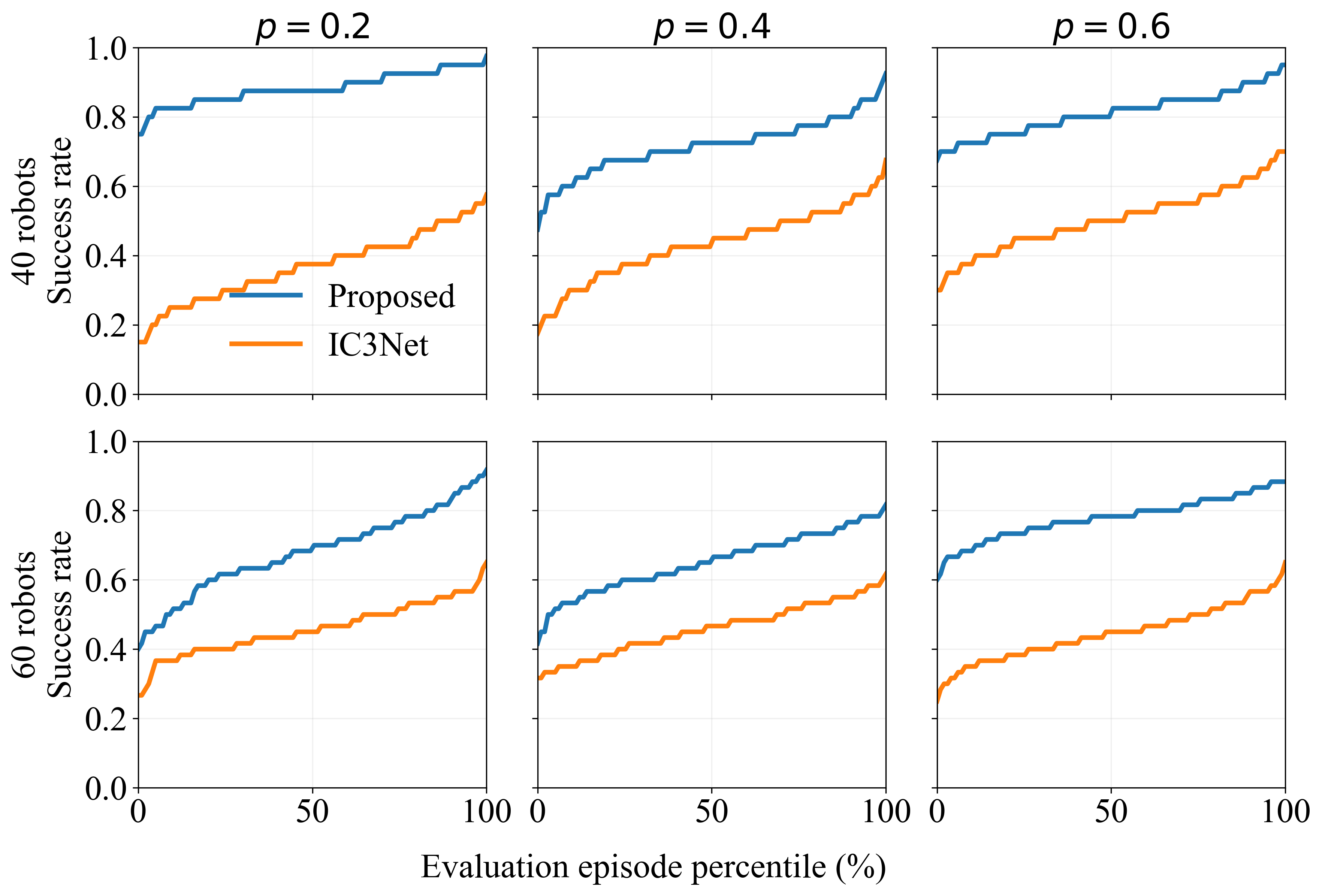}
    \vspace{-0.8cm}
    \caption{Success rate comparison between the proposed method and IC3Net for teams of $40$ and $60$ robots under communication loss probabilities $p\in\{0.2,0.4,0.6\}$.}
    \label{fig:fig2}
    \vspace{-0.7cm}
\end{figure}

Moreover, we examine the communication requirement of the proposed hierarchical architecture. Fig.~\ref{fig:fig3}. (a) compares the communication payload per time step with the fully decentralized all-to-all communication structure. The communication payload is measured as the total number of scalar state and control elements transmitted across communication links in one time step, indicating the required communication bandwidth. As the number of robots increases, the fully decentralized communication load grows rapidly, while the communication load of the proposed method increases approximately linearly.
Fig.~\ref{fig:fig3} compares the measured communication ratio with the analytical result we provide in Section~III.A. The numerical results therefore confirm the communication-scaling analysis: the proposed architecture requires $\mathcal{O}(NZ)$ communication compared to $\mathcal{O}(N^2)$ for robot-level all-to-all communication.

\begin{figure}
    \centering
    \includegraphics[width=\linewidth]{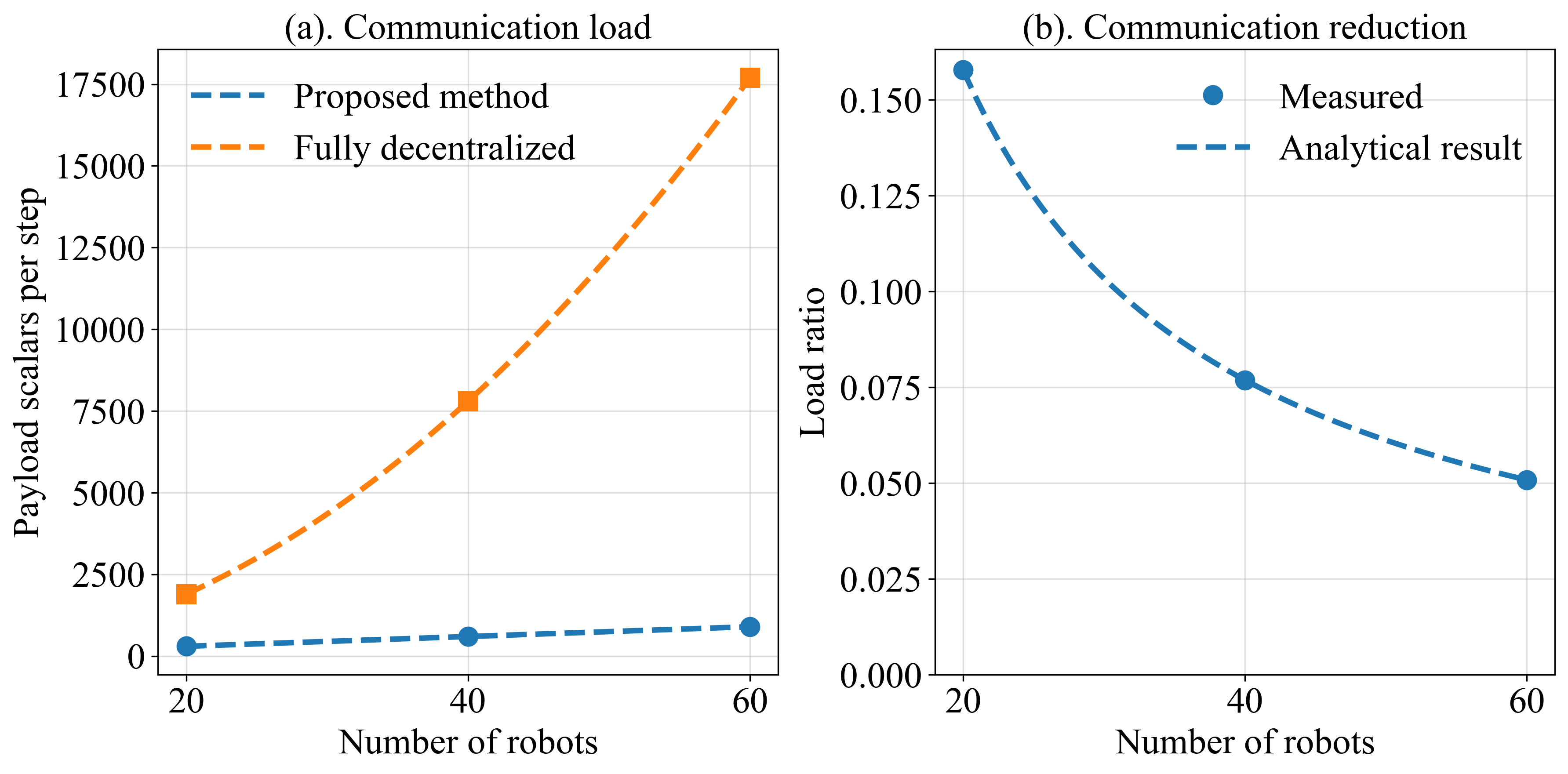}
    \vspace{-0.5cm}
    \caption{(a). Communication load per time step for the proposed method and the fully decentralized baseline. (b). Ratio between payloads of the two methods and the analytical result. }
    \label{fig:fig3}
\end{figure}

In addition, we evaluate safety-related metrics for robot team sizes $40,60,80,90$, and $100$ under a communication loss probability of $p=0.2$. Table~\ref{tab:safety} summarizes the results. No collisions are observed in the evaluation episodes, while the filter modifies fewer than $8\%$ of the proposed robot controls in all tested scenarios. The mean number of searches generally increases for larger teams, indicating that an admissible control can be identified within the normal safety search. These results show that the safety filter prevents observed collisions with limited modification of the learned policy outputs and bounded candidate search.

\begin{table}[t]
\centering
\caption{Safety filter performance under $p=0.2$.}
\label{tab:safety}
\small
\setlength{\tabcolsep}{4pt}
\renewcommand{\arraystretch}{1.0}
\begin{tabular}{c|ccc}
\hline
\textbf{$N$} &
\textbf{Collision rate} &
\textbf{Correction rate} &
\textbf{Mean candidates} \\
\hline
40  & 0.00\% & 3.45\% & 3.79  \\
60  & 0.00\% & 5.48\% & 16.44 \\
80  & 0.00\% & 7.90\% & 24.15 \\
90  & 0.00\% & 6.54\% & 35.63 \\
100 & 0.00\% & 7.34\% & 24.68 \\
\hline
\end{tabular}
\end{table}

\section{Concluding Remarks}
In this paper, we developed a hierarchical MARL framework for warehouse robot coordination under communication loss. We combined recurrent prediction of missing inter-group information, group-level reference and local control policies, and a predictive safety filter for robot team coordination. We analyzed information compression, communication scaling, and safety candidate search. Simulation results demonstrated the method's effectiveness under different configurations. Future work will include considering adaptive and dynamic group assignment, more sophisticated reference space design, and validation in larger and more realistic warehouse environments.

\bibliographystyle{IEEEtran}
\bibliography{reference,ids}

\end{document}